\documentclass[11pt,letterpaper]{article}
\pdfoutput=1
\usepackage[T1]{fontenc}
\usepackage[ascii]{inputenc}
\usepackage{lmodern}
\usepackage[dvipsnames,svgnames]{xcolor}
\usepackage{amsmath,amssymb,amsthm,mathtools,enumitem,mleftright}
\usepackage[margin=0.9in]{geometry}
\usepackage[bookmarksnumbered,linktocpage,hypertexnames=false,colorlinks=true,linkcolor=NavyBlue,urlcolor=NavyBlue,citecolor=ForestGreen,anchorcolor=green,breaklinks=true,pagebackref=false,pdfusetitle]{hyperref}
\usepackage[capitalize,nameinlink]{cleveref}
\usepackage[english]{babel}
\newcommand{\ket}[1]{\lvert #1\rangle}
\newcommand{\bra}[1]{\langle #1\rvert}
\newcommand{\ope}{\mathrm{op}}
\newcommand{\ST}{\mathrm{ST}}
\newcommand{\PST}{P_{\ST}}
\newcommand{\SST}{S_{\ST}}
\newcommand{\mnorm}[1]{\lVert #1\rVert_{\max}}
\newcommand{\transport}{\mathcal{T}}
\newcommand{\qtransport}{\mathcal{C}}
\newcommand{\PCST}{P_{\mathrm{CST}}}
\newcommand{\SCST}{S_{\mathrm{CST}}}
\newcommand{\RR}{\mathbb{R}}
\newcommand{\C}{\mathbb{C}}
\newcommand{\R}{\mathbb{R}}
\newcommand{\KL}{D_{\mathrm{KL}}}
\newcommand{\dTV}{d_{\mathrm{TV}}}
\newcommand{\eps}{\varepsilon}
\newcommand{\ot}{\otimes}
\newcommand{\op}{\oplus}
\newcommand{\ketbra}[2]{\ket{#1}\!\bra{#2}}
\newcommand{\braket}[1]{\langle#1\rangle}
\newcommand{\proj}[1]{\ket{#1}\!\bra{#1}}

\newtheorem*{theorem*}{Theorem}
\newtheorem{theorem}{Theorem}[section]
\crefname{theorem}{Theorem}{Theorems}
\newtheorem{lemma}[theorem]{Lemma}
\crefname{lemma}{Lemma}{Lemmas}
\AddToHook{env/lemma/begin}{\crefalias{theorem}{lemma}}

\crefname{proposition}{Proposition}{Propositions}
\AddToHook{env/proposition/begin}{\crefalias{theorem}{proposition}}

\crefname{corollary}{Corollary}{Corollaries}
\AddToHook{env/corollary/begin}{\crefalias{theorem}{corollary}}
\theoremstyle{definition}
\newtheorem{definition}{Definition}[section]
\crefname{definition}{Definition}{Definitions}
\AddToHook{env/definition/begin}{\crefalias{theorem}{definition}}
\theoremstyle{remark}
\newtheorem*{remark}{Remark}
\crefname{remark}{Remark}{Remarks}
\AddToHook{env/remark/begin}{\crefalias{theorem}{remark}}
\crefname{equation}{Eq.}{Eqs.}
\numberwithin{equation}{section}
\DeclarePairedDelimiter\abs{\lvert}{\rvert}
\DeclarePairedDelimiter\norm{\lVert}{\rVert}

\DeclarePairedDelimiter\parens{\lparen}{\rparen}
\DeclarePairedDelimiter\braces{\lbrace}{\rbrace}

\DeclareMathOperator{\Tr}{Tr}

\DeclareMathOperator*{\argmin}{\arg\min}

\DeclareMathOperator{\diag}{diag}

\begin{document}

\title{Robustness of Hidden-Variable Theories and Matrix Scaling}
\author{Giulio Malavolta\thanks{Bocconi University, \href{mailto:giulio.malavolta@unibocconi.it}{giulio.malavolta@unibocconi.it}}
\and Harold Nieuwboer\thanks{University of Copenhagen, \href{mailto:hani@math.ku.dk}{hani@math.ku.dk}}
\and Akshay Ramachandran\thanks{University of British Columbia, \href{mailto:aramach@cs.ubc.ca}{aramach@cs.ubc.ca}}
\and Michael Walter\thanks{LMU Munich \& MCQST, \href{mailto:michael.walter@lmu.de}{michael.walter@lmu.de}}}
\date{}
\maketitle
\begin{abstract}
Motivated by quantum foundations and complexity theory, Aaronson formalized a hidden-variable theory inspired by a proposal by Schr\"odinger.
To any quantum state and unitary, this Schr\"odinger theory assigns a joint probability distribution via the Sinkhorn algorithm:
rescale the columns and rows of the entrywise modulus of the unitary so that the marginals match the Born rule for the initial and final quantum states, respectively.
He conjectured that this map is \emph{robust}, i.e., inverse polynomially small perturbations in the inputs lead to inverse polynomially small perturbations of the joint distribution, which is important for complexity theoretic applications.
We present a counterexample to this conjecture:
We construct a pure state and two unitaries that are exponentially close, yet the joint probability distributions assigned by the Schr\"odinger theory differ by at least an inverse-linear term in at least one entry, and by a constant in total variation distance.
We also propose a modified version of Schr{\"o}dinger's theory that satisfies robustness, while retaining all of its other desirable properties.
This yields a complete picture of which axioms of Aaronson can be simultaneously satisfied by hidden-variable theories.
\end{abstract}

\section{Introduction}
Quantum mechanics offers a mathematical description of quantum systems in nature:
The state of a system is described by a density matrix~$\rho$ and the time evolution of a closed system for a fixed time interval is described by conjugation with a unitary matrix~$U$.
According to the \emph{Born rule}, the probability of observing the system $\rho$ in the $i$-th basis state is given by~$\rho_{ii}$.
An unintuitive aspect of quantum mechanics is that it does not assign a definite value to a chosen observable at all times (nor to multiple arbitrary observables at the same time).
\emph{Hidden-variable theories} aim to overcome this by introducing additional inaccessible variables at the expense of giving up on other desirable properties of a physical theory.
There is a large body of work on hidden-variable theories, going back all the way to the early days of quantum mechanics; see, e.g., the review by Mermin~\cite{mermin1993hidden}.

Aaronson \cite{Aaronson,AaronsonPreprint,qcsd} proposes an \emph{axiomatic} view of hidden-variable theories: He postulates a series of desirable properties and investigates which of them can be simultaneously satisfied (more on this in \cref{sec:hiddenv}).
Among others, Aaronson studies a hidden-variable theory he calls \emph{Schr\"odinger theory} as it goes back to a remarkable proposal by Schr{\"o}dinger~\cite{Schrodinger1931}: Given a density matrix $\rho$ and a unitary matrix~$U$, a joint distribution is computed by rescaling the rows and columns of~$\abs U=(\abs{U_{ji}})_{ji}$ so that the column marginals coincide with $\rho_{ii}$ and the row marginals with $(U\rho U^\dagger)_{jj}$; one then obtains a stochastic transition map by dividing by the column marginals (with a suitable limiting procedure for zero marginals).

As a main open problem, Aaronson asks whether this Schr{\"o}dinger theory is \emph{robust}.
This question can be understood in terms of stability of the Sinkhorn algorithm:%
\footnote{In fact, the question is more delicate than ordinary stability, since the input matrix and the target row and column marginals are correlated in the application in the Schr\"odinger theory.}
Starting from two inputs $(U_0,\rho_0)$ and $(U_1, \rho_1)$ that are close to each other, will the corresponding scaled matrices also be close to each other?
In other words, is the map induced by the Sinkhorn algorithm robust to small perturbations of the input?
Or can it output matrices that have distance exponentially larger in the dimension?
Besides being a question of independent interest on matrix scaling, Aaronson shows that, if one could prove the desired robustness for Schr{\"o}dinger's theory, then one would obtain a complete characterization of which of his axioms can be simultaneously satisfied.

In this paper, we (1)~disprove Aaronson's conjecture on the robustness of Schr\"odinger theory, (2)~show that a natural modification of the theory salvages robustness, while still satisfying symmetry, indifference, and distributional product commutativity, and (3)~show new limitations on combinations of axioms.
Together, this gives a complete picture as to which combinations of Aaronson's axioms can be simultaneously satisfied.

\subsection{Background on Aaronson's Axioms for Hidden-Variable Theories}\label{sec:hiddenv}
In Aaronson's formulation, a \emph{hidden-variable theory} is a map~$S$ that, on input a density matrix~$\rho$ and a unitary matrix~$U$ in any dimension~$N$, outputs a column-stochastic \emph{transition matrix}~$S(\rho, U)$.
Define the corresponding joint distribution by $P(\rho, U) \coloneqq S(\rho, U)\cdot\mathrm{diag}(\rho_{11}, \dots, \rho_{NN})$.
We require~that:
\begin{equation}\label{eq:marginal}
    \sum_{j} P(\rho, U)_{ji} \stackrel!= \rho_{ii} \qquad\text{and}\qquad
    \sum_{i} P(\rho, U)_{ji} = \sum_{i} S(\rho, U)_{ji} \cdot \rho_{ii} \stackrel!= (U\rho U^\dagger)_{jj}.
\end{equation}
In other words, the marginals of the joint distribution should equal the \emph{Born marginals} (defined as the collection of diagonal entries) of the initial and final states.
The second condition in~\eqref{eq:marginal} is a genuine requirement; the first condition already follows from the stochasticity of the transition matrix.
Aaronson considers the following axioms as desirable properties for a hidden-variable theory~\cite{Aaronson,AaronsonPreprint}:
\begin{enumerate}
\item \emph{Symmetry:}
This axiom states that the theory is invariant under relabeling the basis states. Formally,
$S(\rho,U) = Q S(Q^{-1}\rho Q,Q^{-1}UQ) Q^{-1}$
for all $\rho,U$ and all permutation matrices~$Q$.
\item \emph{Indifference:}
This axiom states that $S(\rho,U)$ preserves the block structure of~$U$.
In particular, if $U$ is block diagonal then so must be $S(\rho, U)$, with the same block structure (or a refinement thereof).
Formally, a pair~$(I,J)$ of nonempty subsets of $\{1,\dots,N\}$ is called a \emph{block} of an $N \times N$ matrix~$M$ if $M_{ji}=0$ whenever~$i\in I, j\not\in J$ or~$i\not\in I, j\in J$ (for unitaries, this implies that~$\abs I = \abs J$).
The indifference axiom requires that any block of $U$ is also a block of~$S(\rho, U)$.%
\footnote{Here we adopt the definition of indifference in \cite{AaronsonPreprint}.
In~\cite{Aaronson}, this is only required for ``diagonal'' blocks, i.e., blocks of the form~$(I,I)$.
The product theory does not satisfy this weaker requirement either.
\Cref{thm:scott a,thm:scott b,thm:decompinvcom,thm:productcom} remain valid with this weaker notion of indifference.}

This can also be phrased in terms of so-called minimal blocks.
A block~$(I,J)$ of~$U$ is \emph{minimal} if no other block $(I',J')$ has $I' \subsetneq I$ or $J' \subsetneq J$.
The set of all minimal blocks of $U$ partitions the column and row indices.
A theory satisfies the indifference axiom if and only if $S(\rho, U)_{ji}=0$ whenever~$i$ and~$j$ belong to different minimal blocks of $U$.

\item \emph{Robustness:} This axiom states that, for every positive polynomial~$r$, there is a positive polynomial~$s$ such that the following holds for all pairs of density matrices~$\rho_0$, $\rho_1$ and unitaries~$U_0$, $U_1$ in every dimension~$N$:
\[
 \mnorm{\rho_0-\rho_1},\ \mnorm{ U_0-U_1}
 \le \frac1{s(N)}
 \quad\Longrightarrow\quad
 \mnorm{P(\rho_0, U_0)-P(\rho_1,U_1)}
 \le \frac1{r(N)}.
\]
Here, we use the entrywise maximum norm $\mnorm{M} \coloneqq \max_{j,i} \abs{M_{ji}}$.

\emph{Block Robustness:} This is a weaker variant of the above axiom.
It requires the above condition to hold only if $U_0$ and $U_1$ have the same minimal blocks (defined in the preceding point).

\item \emph{Decomposition Invariance:}
This axiom states that for any convex decomposition of $\rho$ into pure states, $\rho=\sum_k\lambda_k \ket{\psi_k}\bra{\psi_k}$, and for every unitary~$U$, it holds that
\[
     S(\rho,U) = \sum_k \lambda_k S(\ket{\psi_k}\bra{\psi_k},U).
\]

\item \emph{Commutativity:}
Let $A,B$ be positive integers and $N=AB$.
Identify $\C^N = \C^A \ot \C^B$ using the lexicographically-ordered product basis (so $\ot$ coincides with the Kronecker product of matrices).
For every state $\rho$ and any pair of local unitaries~$\mathcal U=U_A\otimes I_B$ and~$\mathcal V=I_A\otimes U_B$, the commutativity axiom requires that
  \begin{equation}\label{eq:oldcomm}
 S(\mathcal U\rho\mathcal U^\dagger,\mathcal V)S(\rho,\mathcal U)
 =S(\mathcal V\rho\mathcal V^\dagger,\mathcal U)S(\rho,\mathcal V).
 \end{equation}

\emph{Product Commutativity:}
This is a weaker variant of the above axiom.
It requires \cref{eq:oldcomm} to hold only for pure product states, i.e., $\rho = \ket\psi\bra\psi$ where $\ket\psi = \ket{\psi_A} \ot \ket{\psi_B}$.
\end{enumerate}
We also consider natural weakenings of the above commutativity properties, defined by comparing the joint distributions rather than the transition matrices.
They are implicit in Aaronson's work; indeed some of his theorems apply with this weaker condition instead (see below and \cref{tab:comparison}).
\begin{enumerate}
\item[5'.] \emph{Distributional Commutativity}  and \emph{Distributional Product Commutativity:}
These are defined as above, but with \cref{eq:oldcomm} replaced by the analogous condition for the joint distributions:
\begin{equation}\label{eq:propercomm}
    S(\mathcal U\rho\mathcal U^\dagger,\mathcal V)
    P(\rho,\mathcal U)
= S(\mathcal V\rho\mathcal V^\dagger,\mathcal U)
    P(\rho,\mathcal V).
\end{equation}
This condition is equivalent to~\cref{eq:oldcomm} if~$\rho_{ii}>0$ for all~$i$, but otherwise weaker, as it is obtained from~\cref{eq:oldcomm} by multiplication with the matrix $\mathrm{diag}(\rho_{11}, \dots, \rho_{NN})$.
\end{enumerate}

Aaronson discusses several constructions of hidden-variable theories and the axioms that they satisfy.
A simple \emph{Product Theory} satisfies symmetry, robustness, decomposition invariance and commutativity, but not indifference.
A variant by Dieks~\cite{Dieks1994Modal} restores indifference, but only satisfies block robustness and distributional product commutativity; here we observe that it is not decomposition invariant (\cref{thm:decompinvcom}).
A \emph{Flow Theory} satisfies symmetry, indifference and robustness, but neither decomposition invariance nor commutativity.
Finally, Aaronson proposes the aforementioned \emph{Schr\"odinger Theory}.
He proves that it satisfies symmetry, indifference, and distributional product commutativity, but not decomposition invariance; robustness is left as a main open problem.

None of these theories satisfy all desired axioms simultaneously.
To explain this, Aaronson also proves the following fundamental limitations.

\begin{theorem}[{\cite[Thm.~1]{Aaronson}}]
\label{thm:scott a}
No hidden-variable theory satisfies both indifference and distributional commutativity.
\end{theorem}

\begin{theorem}[{\cite[Thm.~2~(i)]{Aaronson}}]
\label{thm:scott b}
No hidden-variable theory satisfies indifference, robustness, and decomposition invariance.
\end{theorem}

Aaronson also comments on complexity theoretic implications of hidden-variable theories.
Specifically, he shows that the graph isomorphism problem can be solved with bounded error by a deterministic classical polynomial-time Turing machine that makes a single query to an oracle that samples histories according to an arbitrary indifferent and robust hidden-variable theory \cite{Aaronson}.
Formally, the latter defines a complexity class \textsc{DQP} (dynamical quantum polynomial time) which he shows contains the class \textsc{SZK} (statistical zero knowledge). Follow-up works \cite{AaronsonBFL16,AaronsonGIMR25,MiloschewskyP25} define and study the related complexity class \textsc{PDQP} (product dynamical quantum polynomial time) characterizing the power of quantum computers which can make \emph{non-collapsing} measurements.

\begin{table}[t]
    \centering
    \begin{tabular}{c|c|c|c|c|c}
       & Product & Dieks & Flow & Schr\"odinger & \textbf{Capped Schr\"odinger} \\
       \hline
       Symmetry & $\checkmark$ & $\checkmark$ & $\checkmark$ & $\checkmark$ & $\boldsymbol\checkmark$ \\
       Indifference & no & $\checkmark$ & $\checkmark$ & $\checkmark$ & $\boldsymbol\checkmark$ \\
       Robustness & $\checkmark$ & block & $\checkmark$ & \textbf{no} & $\boldsymbol\checkmark$ \\
       Decomposition Invariance & $\checkmark$ & \textbf{no} & no & no & \textbf{no} \\
       Commutativity & $\checkmark$ & product\textbf{$^\ast$} & no & product\textbf{$^\ast$} & \textbf{product$^\ast$}
    \end{tabular}
    \caption{Comparison of different hidden-variable theories and the axioms that they satisfy.
    Asterisks refer to the weaker \emph{distributional} product commutativity property.
    Results obtained in this paper are highlighted in bold.}
    \label{tab:comparison}
\end{table}

\subsection{Result: Schr\"odinger Theory is Not Robust}
\label{sec:result st not robust}

The first result of this paper resolves Aaronson's question in the negative:
we show that the Schr\"odinger theory is not robust.
Let us recall its definition.
Given a unitary matrix $U$ and a density matrix $\rho$, we can define a joint distribution $\PST(\rho,U)$ with column marginals $p_i = \rho_{ii}$ and row marginals $q_j = (U\rho U^\dagger)_{jj}$ as the limit of the following iteration, which is known as \emph{Sinkhorn's algorithm}~\cite{Sinkhorn,SinkhornKnopp}:
Starting from $X^{(0)}\coloneqq\abs{U}\coloneqq(\abs{U_{ji}})_{j,i}$, define the alternating column and row scalings
\begin{align*}
X^{(2t+1)}_{ji}
\coloneqq
X^{(2t)}_{ji}
\frac{p_i}{\sum_k X^{(2t)}_{ki}},\qquad
X^{(2t+2)}_{ji}
\coloneqq
X^{(2t+1)}_{ji}
\frac{q_j}{\sum_k X^{(2t+1)}_{jk}},
\end{align*}
with the convention $0/0\coloneqq0$.
The joint distribution is defined as $\PST(\rho,U)\coloneqq \lim_{t\to\infty}X^{(t)}$.
The fact that this limit exists is not immediate, but we defer this discussion for now (more on matrix scaling in \cref{subsec:intro capped schrodinger robust,sec:preliminaries}).
The Schr\"odinger theory is then defined via
$\SST(\rho,U)_{ji} \coloneqq \PST(\rho,U)_{ji} / p_{i}$ when all~$p_i \neq 0$, and otherwise as the limit~$\lim_{\eps \downarrow 0} \SST((1-\eps) \rho + \eps \frac{I_N}{N}, U)$.%
\footnote{There are other reasonable choices for dealing with columns with~$p_i=0$.
For instance, one may choose~$\SST(\rho,U)_{ji} \coloneqq \abs{U_{ji}}^2$ when~$p_i=0$ instead, and we shall do so later for our modification of the Schr\"odinger theory.
This does not affect the robustness discussion, as the latter only refers to the joint distribution~$P$ rather than the transition matrix~$S$.}
In either case, $\PST(\rho,U) = \SST(\rho,U) \diag(p_1,\dots,p_N)$, so the former is indeed the joint distribution arising from the latter transition matrix.
We show:


\begin{theorem}\label{thm:no go}
    Schr{\"o}dinger's hidden-variable theory is not block-robust.
    In particular, it is not robust.
\end{theorem}

\noindent
To prove this result, we construct for any $N = 2n+1 \geq 5$ an $N\times N$ density matrix~$\rho$ and two unitaries~$U_0$ and~$U_1$ such that
\[
 \mnorm{U_0-U_1}<2^{-n},\qquad
 \mnorm{\PST(\rho,U_0)-\PST(\rho,U_1)}
 >\frac1{10N}.
\]
Our $U_0$ and $U_1$ have the same support, so the non-robustness is caused by the scaling procedure itself, rather than an artifact of a change in support.

\subsection{Result: Capped Schr\"odinger Theory is Robust}\label{subsec:intro capped schrodinger robust}
We present a new hidden-variable theory, a variant of Schr\"odinger theory, that salvages robustness.
Before giving the precise definition, we provide some additional context to motivate the definition.
It is well-known that if there exists an asymptotic row- and column-rescaling~$P_\star$ of a non-negative matrix~$Q$ to prescribed marginals~$r,s \in \RR_{\geq 0}^N$, then it is the \emph{minimum information projection}~\cite{IrelandKullback1968,csiszarIDivergenceGeometryProbability1975} of the matrix~$Q$ onto the polytope of non-negative matrices~$P$ with these marginals.
In fact, Schr\"odinger~\cite{Schrodinger1931} originally motivated his theory through the maximum-entropy-principle, which is essentially equivalent to this minimum-information viewpoint.
Formally, $P_\star$ is the unique minimizer of the \emph{Kullback--Leibler divergence}~$\KL(P \Vert Q) \coloneqq \sum_{j,i} P_{ji} \ln (P_{ji}/Q_{ji})$, where~$P$ is taken from the~polytope
\[
  \transport(r,s) \coloneqq \left\{ P \in \RR_{\geq 0}^{N \times N} : \sum_{j} P_{ji} = r_i, \, \sum_{i} P_{ji} = s_j \text{ for all }j, i\right\}.
\]
Conversely, Sinkhorn's algorithm converges to~$P_\star$ whenever $\KL(P \Vert Q) < \infty$ for some $P$ in this polytope; that is, when there exists a matrix $P$ with the desired marginals with support inside the support of~$Q$.
In particular, $P_\star$ is always a limit of row- and column-rescalings of~$Q$.
See \cref{sec:preliminaries} and in particular \cref{thm:sinkhorn-limit-as-KL-projection} for more on this perspective on matrix scaling.
The above discussion reveals that Aaronson's Schr\"odinger theory admits the following alternative formulation:
\[ \PST(\rho,U)
=
\argmin_{P\in\transport(p,q)}
\KL(P \,\Vert\, \abs{U}).
\]
To define our \emph{capped Schr\"odinger theory}, we modify the above by imposing an extra \emph{capacity constraint}, motivated by the fact that the elements of $\transport(r,s)$ can be interpreted as flows (more on this below).
Given a density matrix $\rho$ and a unitary $U$, let
\[
\qtransport(\rho,U)\coloneqq
\left\{
P\in\mathbb R_{\geq0}^{N\times N}:
\sum_j P_{ji}=\rho_{ii},\
\sum_i P_{ji}=(U\rho U^\dagger)_{jj},\
P_{ji}\leq C_{ji}\ \text{for all }j,i
\right\}
\]
where $C_{ji}\coloneqq 2 \abs{U_{ji}}\sqrt{\rho_{ii}(U\rho U^\dagger)_{jj}}$.
We define a joint distribution by minimizing the Kullback--Leibler divergence over~$\qtransport(\rho,U)$:
\[
\PCST(\rho,U)
\coloneqq
\argmin_{P\in\qtransport(\rho,U)}
\KL(P \,\Vert\, \abs{U}).
\]
We establish that $\qtransport(\rho,U)$ is always non-empty; uniqueness of the minimizer then follows from strict convexity of the objective.
Note that~$\PCST(\rho,U)$ agrees with~$\PST(\rho,U)$ whenever the latter satisfies the capacity constraint, but in general they may differ.
Finally, we define the \emph{capped Schr\"odinger theory} by~$\SCST(\rho, U)_{ji} = \PCST(\rho,U)_{ji} / \rho_{ii}$ whenever~$\rho_{ii} \neq 0$, and~$\SCST(\rho,U)_{ji} = \abs{U_{ji}}^2$~otherwise.

\begin{theorem}\label{thm:yay}
The capped Schr\"odinger hidden-variable theory satisfies symmetry, indifference, robustness, and distributional product commutativity.
\end{theorem}

We comment briefly on the proof of robustness, which is the most difficult aspect of \cref{thm:yay}.
The main advantage of the capacity bound is that it enables us to project elements of $\qtransport(\rho,U)$ to nearby elements of $\qtransport(\tilde\rho,\tilde U)$ when $(\rho,U)\approx(\tilde\rho,\tilde U)$.
This allows us to show that $\PCST(\rho,U)$ is near an approximate minimizer for the optimization problem defining $\PCST(\tilde\rho,\tilde U)$, and vice versa.
By using a strong convexity property of the KL-divergence, we can then deduce that $\PCST(\rho,U) \approx \PCST(\tilde\rho,\tilde U)$.

\medskip

As noted earlier, the elements of $\transport(r,s)$ can be interpreted as flows of a bipartite graph with $N$ source vertices and $N$ sink vertices that send $r_i$ units of flow from the $i$-th source and receive $s_j$ units of flow at the $j$-th sink.
The polytope $\qtransport(\rho,U)$ can be interpreted analogously, but with additional capacity constraints on the edges $i\to j$.
Aaronson's Flow Theory $P_{\mathrm{FT}}$ is defined through solving maximum-flow problems on a polytope similar to~$\qtransport(\rho,U)$, but using the capacity $C_{ji} \coloneqq \abs{U_{ji}}$~\cite{Aaronson}.
If one optimizes the Kullback--Leibler divergence over this polytope, one also obtains a robust theory.
However, unlike our capped Schr\"odinger theory, the resulting theory does \emph{not} satisfy the distributional product commutativity property.
This is related to the fact that the resulting optimization problem for product states and local unitaries will not factorize appropriately.
Aaronson~\cite{Aaronson} also notes that for smaller capacities $C_{ji} \coloneqq  |U_{ji}|^{1+\eps}$, the flow polytope can be empty in general.
One can observe that for our capacity constraint $C_{ji} \coloneqq  2 \abs{U_{ji}} \sqrt{q_{j} p_{i}}$, the constant $2$ is also tight for feasibility.
Indeed one can take $U_\eps = \begin{bsmallmatrix} \sqrt{\eps} & \sqrt{1-\eps} \\ \sqrt{1-\eps} & -\sqrt{\eps} \end{bsmallmatrix}$ and $\rho = \ket{+}\bra{+}$.
When $\eps \to 0$, $p_1+q_1-1 = \sqrt{\eps (1-\eps)}$, whereas $q_1 \to \frac12$ as $\eps \to 0$, and so $C_{11} = 2 \sqrt{\eps q_1/2}$ scales as~$\sqrt{\eps}$ when~$\eps \to 0$; using a smaller constant than~$2$ in the definition of~$C_{11}$ would violate (for sufficiently small~$\eps>0$) the inequality~$p_1+q_1-1 \leq C_{11}$, which is necessary for the non-emptiness of~$\qtransport(\rho,U)$.

\subsection{Result: New Limitations on Axioms for Hidden-Variable Theories}\label{subsec:intro new limitations}
To obtain a complete picture of which axioms can be simultaneously satisfied, we also identify the following new limitations.
They are proved in \cref{sec:limitations}.

\begin{theorem}\label{thm:productcom}
No hidden-variable theory satisfies both indifference and product commutativity (in the stronger sense of \cref{eq:oldcomm}).
\end{theorem}

\noindent
This theorem complements \cref{thm:scott a}.
The following theorem is a strengthening of \cref{thm:scott b}.
In particular, it implies that Dieks' theory cannot be decomposition-invariant, since it is indifferent.

\begin{theorem}\label{thm:decompinvcom}
No hidden-variable theory satisfies both indifference and decomposition invariance.
\end{theorem}

The above results give a complete picture of which combinations of Aaronson's axioms are possible.
All axioms but indifference can be achieved simultaneously, by the product theory, as already observed in \cite{Aaronson}.
If indifference is desired, we must give up decomposition invariance (\cref{thm:decompinvcom}), and commutativity can at best hold in the sense of distributional product commutativity (\cref{thm:scott a,thm:productcom}); the latter is achieved by our capped Schr\"odinger theory along with symmetry and robustness (\cref{thm:yay}).
See also \cref{tab:comparison}.

\subsection{Statement on AI Use}
The counterexamples and the capped version of Schr{\"o}dinger's theory, along with its proof of robustness, were found by ChatGPT in interactive conversations with the authors.
This manuscript was entirely written by the authors, who take full responsibility for its content.

\subsection{Related Work}
Sinkhorn's algorithm~\cite{Sinkhorn,SinkhornKnopp} has been independently discovered and re-discovered in many different settings, e.g. in estimating dynamics in telephone traffic~\cite{kruithofTelefoonverkeersrekening1937}, for contingency table estimation~\cite{demingLeastSquaresAdjustment1940}, and a slight modification was used to give a direct existence proof for a solution to Schr\"odinger's original problem~\cite{fortetResolutionDunSysteme1940}.
As a consequence, it goes by various other names: \emph{iterative proportional fitting} in statistics, or the \emph{RAS algorithm} in economics.
The matrix scaling problem can also be viewed as a discrete version of the \emph{Schr\"odinger bridge} problem.

As a result, there is a large body of work, and various generalizations and applications are known.
For instance, it can be used to solve entropically-regularized optimal transport problems~\cite{cuturiSinkhornDistancesLightspeed2013,altschuler2017nearlinear}.
Matrix scaling is also done in practice for preconditioning linear systems~\cite{lapack}.
In statistics, it gives minimum-information estimates of contingency tables when the marginals are known~\cite{IrelandKullback1968,csiszarIDivergenceGeometryProbability1975}, as well as being useful for maximum-likelihood estimation in log-linear models~\cite{darrochGeneralizedIterativeScaling1972}.
In computer science, it can be used to obtain an exponential factor approximation to the permanent of a nonnegative matrix~\cite{linial1998deterministic}.
For an overview of matrix scaling and its applications, we refer to the survey~\cite{idelReviewMatrixScaling2016}.
There are non-commutative generalizations, such as operator scaling~\cite{gurvitsClassicalComplexityQuantum2004,gargDeterministicPolynomialTime2016,gargOperatorScalingTheory2020} and tensor scaling~\cite{burgisserEfficientAlgorithmsTensor2018,burgisserAlternatingMinimizationScaling2018,gargOperatorScalingTheory2020}, which is related to a variant of the quantum marginal problem and to the classification of multipartite entanglement~\cite{klyachkoQuantumMarginalProblem2004,walterEntanglementPolytopesMultiparticle2013}; these also have applications in statistics, through matrix- and tensor-normal models~\cite{franks2020rigorous,amendola2024bridge,franksOptimalSampleComplexity2026}.
There is also a variant of matrix scaling for unitary matrices~\cite{idelSinkhornNormalForm2015} with an application to quantum circuit decompositions.

As a result of this widespread interest, there are various different robustness-type bounds available for Sinkhorn scaling and consequently, on Schr\"odinger theory \cite{sinkhornContinuousDependenceA1972,CarlierLaborde2020,ghosalStabilityEntropicOptimal2022,EcksteinNutz2022,Landa2022,NutzWiesel2023,DeligiannidisEtAl2024,DuanLyuPowell2026}.
For fixed dimension and fixed marginals, on the domain where the scaling limit exists, it is continuous~\cite{sinkhornContinuousDependenceA1972} as a function of the matrix that is being scaled.
Some authors give quantitative bounds in the setting where one perturbs the prescribed row and column marginals, while keeping the matrix being scaled fixed.
Results that also allow the matrix to vary either require additional non-degeneracy assumptions ensuring that all entries are positive (which naturally occurs in the optimal transport setting), or give bounds whose constants deteriorate when positive matrix entries become very small or when rows have little common support.
These conditions are indeed consistent with our counterexample, which considers sparse matrices with asymptotically small entries.

We also note that entropically-regularized optimal transport with capacity constraints, of which the definition of~$\PCST$ is a special case, has been studied before~\cite[Sec.~5.2]{benamouIterativeBregmanProjections2015}, but we are not aware of any previous results on stability of the solution in this setting.

\section{Preliminaries}\label{sec:preliminaries}
We write $[N] \coloneqq \{1,\dots,N\}$.
For vectors and matrices alike, we use $\norm{\cdot}_1$ for the entrywise $\ell^1$ norm.
For matrices $X$, we also use the entrywise max-norm, defined by $\mnorm{X} \coloneqq \max_{j,i}|X_{ji}|$.
For~$P,Q \in \RR^{N \times N}_{\geq 0}$, we define their \emph{total variation distance} as
\[
 \dTV(P,Q) \coloneqq \frac12\sum_{j,i}|P_{ji}-Q_{ji}|,
\]
and their \emph{Kullback--Leibler (KL) divergence} as
\[
 \KL(P\Vert Q) \coloneqq \sum_{j,i} P_{ji} \ln \frac{P_{ji}}{Q_{ji}} \in \RR \cup \{\infty\}.
\]
By convention, we take $0 \ln \frac{0}{q} \coloneqq 0$ for any $q \geq 0$, and we set $p \ln \frac p0 \coloneqq \infty$ for any $p>0$.
Thus, $\KL(P\Vert Q)<\infty$ if and only if the support of~$P$ is contained in the support of~$Q$.
We emphasize that in both these definitions, we do not require $P,Q$ to be normalized probability distributions.

In the proof of robustness of the capped Schr\"odinger theory, we will use a strong convexity property of the KL divergence.
To state it formally, let $C \in \RR^{N \times N}_{\geq0}$ (not necessarily normalized) and define the probability simplex $\Delta_C \coloneqq \{ P \in \RR_{\geq0}^{N \times N} : \sum_{j,i} P_{ji}=1 \text{ and } C_{ji}=0 \Rightarrow P_{ji}=0 \; \forall j,i \}$.
Then, $\KL(\cdot \Vert C)$ is $1$-strongly convex with respect to the $\ell^1$-distance on~$\Delta_C$:
for any $P, P' \in \Delta_C$ and $t \in (0,1)$, with $P_t \coloneqq (1-t) P + t P'$, we have%
\footnote{We recall the standard argument: The left-hand side is equal to $(1-t) \KL(P \Vert P_t) + t \KL(P' \Vert P_t)$. Thus the inequality follows from Pinsker's inequality, which states that $\KL(P \Vert  Q) \geq \frac12 \norm{P-Q}_1^2$ for any two probability distributions~$P,Q$.}
\begin{align}\label{eq:strong convex}
    (1-t) \KL(P \Vert C) + t \KL(P' \Vert C) - \KL(P_t \Vert C)
& \geq \frac{(1-t)t}2 \, \norm{P - P'}_1^2.
\end{align}

\begin{definition} [Matrix Scaling Problem]
    Given non-negative $Q \in \RR^{N \times N}_{\geq 0}$, desired marginals $r,s \in \RR^{N}_{\geq 0}$, and precision $\delta\geq0$, find $x,y \in \RR^{N}$ such that the \emph{scaling} $P = (P_{ji} \coloneqq  e^{x_{i}} Q_{ji} e^{y_{j}})_{j,i \in [N]}$ satisfies
    \begin{equation}\label{eq:delta scaling}
        \| P^\top 1_{N} - r \|_{1} + \|P 1_{N} - s \|_{1} \leq \delta .
    \end{equation}
    The instance is called \emph{scalable} if there exists a solution for $\delta = 0$; and is called \emph{asymptotically scalable} if there exists a solution for every $\delta > 0$.
\end{definition}

An important feature is that if~$Q$ is asymptotically scalable, with~$P(\delta)$ approximate rescalings, then there is a unique limit~$P_\star = \lim_{\delta \to 0} P(\delta)$~\cite{IrelandKullback1968,csiszarIDivergenceGeometryProbability1975}.
In the case where~$Q$ is scalable, this was established in increasing levels of generality: in~\cite{Sinkhorn} for entrywise positive matrices and uniform target marginals, in~\cite{sinkhornDiagonalEquivalenceMatrices1967} for entrywise positive and arbitrary target marginals, in~\cite{brualdiDiagonalEquivalenceNonnegative1966,SinkhornKnopp} for (indecomposable) non-negative matrices and uniform target marginals, and in~\cite{menonMatrixLinksExtremization1968,brualdiDADTheoremArbitrary1974} for (indecomposable) non-negative matrices and non-uniform target marginals.
In the general asymptotically scalable case with arbitrary support and target marginals, uniqueness follows from~\cite{csiszarIDivergenceGeometryProbability1975}, through the formulation in terms of KL divergence optimization, as briefly discussed in the introduction.

We will now review the optimization viewpoint in some more detail; see also \cite{LeeNotes}.
To this end, we define the transportation polytope, which consists of saturating flows for the following network.
We have a single source vertex~$S$ and target vertex~$T$, input vertices~$i$, and output vertices~$j$.
There are edges~$S \to i$, $j \to T$, $i \to j$ with capacities~$r_i$, $s_j$, $C_{ji}$ respectively.
Then the \emph{transportation polytope}~$\transport(r,s|C)$ is defined by
\begin{equation}\label{eq:transpoly}
  \transport(r,s|C) = \left\{ P \in \RR_{\geq 0}^{N \times N} \;:\; \sum_j P_{ji} = r_i, \, \sum_i P_{ji} = s_j, \, P_{ji} \leq C_{ji} \text{ for all } j,i \right\}.
\end{equation}
We similarly define~$\transport(r,s)$ but without the capacity constraints on the edges~$i\to j$:
\[
  \transport(r,s) = \left\{ P \in \RR_{\geq 0}^{N \times N} \;:\; \sum_j P_{ji} = r_i, \, \sum_i P_{ji} = s_j, \text{ for all } j,i \right\}.
\]
The connection to the matrix scaling problem is then through the following theorem:

\begin{theorem}[\cite{csiszarIDivergenceGeometryProbability1975}]
\label{thm:sinkhorn-limit-as-KL-projection}
    Let~$Q \in \RR_{\geq 0}^{N \times N}$ and~$r, s \in \RR_{\geq 0}^N$.
    Then~$Q$ is asymptotically scalable to marginals~$(r,s)$ if and only if there exists $P \in \transport(r,s)$ such that $P_{ji}=0$ whenever $Q_{ji}=0$, or equivalently $\KL(P\Vert Q) < \infty$.
    In this case, the convex program
    \[
        \min_{P \in \transport(r,s)} \KL(P \Vert  Q)
    \]
    has a unique minimizer~$P_\star$.
    It is the unique limit of any sequence of $\delta$-$(r,s)$-scalings of~$Q$ with~$\delta \to 0$.
    Moreover, the iterates of the Sinkhorn scaling algorithm converge to~$P_\star$.
\end{theorem}
We shall use the following incarnation of the max-flow min-cut theorem later on, to establish non-emptiness of the polytope~$\qtransport(\rho, U)$ used for defining the capped Schr\"odinger theory:
\begin{theorem}[\cite{fordFulkerson1956,eliasFeinsteinShannon1956}]
  \label{thm:max-flow min-cut}
  Let~$C \in \RR_{\geq 0}^{N \times N}$ and~$r,s \in \RR_{\geq0}^N$ with~$\sum_i r_i = \sum_j s_j = V$.
  The transportation polytope~$\transport(r,s|C)$ is non-empty if and only if for every~$I,J \subseteq [N]$,
  \[
     r(I) + s(J) \leq V + C(I,J)
  \]
  where $r(I) \coloneqq \sum_{i \in I} r_i$, $s(J) \coloneqq \sum_{j \in J} s_j$, and $C(I,J) \coloneqq \sum_{i\in I, j \in J} C_{ji}$.
\end{theorem}

\section{Counterexample to Robustness of Schr\"odinger Theory}\label{sec:counterexample}
We first present a counterexample to the robustness of Sinkhorn scaling.
Then we lift it to obtain a counterexample to robustness of the Schr\"odinger hidden-variable theory by constructing corresponding unitaries and states, establishing \cref{thm:no go}.

For $n\geq 2$ and $0<\delta<1$, let
\[
 M_\delta=
 \begin{pmatrix}
 1&1&&&\\
 &1&1&&\\
 &&\ddots&\ddots&\\
 &&&1&1\\
 \delta&&&&1
 \end{pmatrix}\in\mathbb R_{\ge0}^{n\times n},
\]
where blank entries are zero.

\begin{lemma}\label{lem:scale}
The (unique) doubly stochastic scaling of $M_\delta$ is given by
\[
 T_\delta=
 \begin{pmatrix}
 t_\delta&1-t_\delta&&&\\
 &t_\delta&1-t_\delta&&\\
 &&\ddots&\ddots&\\
 &&&t_\delta&1-t_\delta\\
 1-t_\delta&&&&t_\delta
 \end{pmatrix},
 \qquad
 t_\delta=\frac{1}{1+\delta^{1/n}}.
\]
\end{lemma}

\begin{proof}
Matrix scaling will never change the support of the starting matrix~$M_\delta$. 
On the other hand, any doubly stochastic matrix with the same support as~$M_\delta$ must have a common diagonal entry~$t_\delta\in(0,1)$ and a common cyclic off-diagonal entry~$1-t_\delta$: for~$1\leq i<n$, comparing row~$i$ with column~$i+1$ forces successive diagonal entries to agree.
Thus, if~$M_\delta$ has a doubly stochastic scaling, it must have the displayed form of~$T_\delta$ for some value of~$t_\delta$.
To determine the latter, we use that the cross-ratio
\[
 \chi(M)=\frac{M_{1,2}M_{2,3}\cdots M_{n-1,n}M_{n,1}}
 {M_{1,1}M_{2,2}\cdots M_{n,n}}
\]
is invariant under row and column scaling (indeed, both the numerator and the denominator change by the same factor: the product of all scaling factors).
Thus, if $T_\delta$ is a scaling of $M_\delta$, we must have
\begin{align*}
    \delta = \chi(M_\delta) = \chi(T_\delta) = \parens*{ \frac{1-t_\delta}{t_\delta} }^n,
\end{align*}
which gives the stated value of~$t_\delta$.
Finally, we observe that $T_\delta$ is indeed a scaling of $M_\delta$:
\begin{equation*}
    \diag(1,\delta^{-1/n},\dots,\delta^{-(n-1)/n}) M_\delta \diag(t_\delta, t_\delta \delta^{1/n}, \dots, t_\delta \delta^{(n-1)/n})
= T_\delta.
\end{equation*}
This concludes the proof.
\end{proof}

Taking $\delta \coloneqq  2^{-n}, \delta' \coloneqq  2^{-3n}$, we see that an exponentially small perturbation leads to a change of constant size in the scaling limit:
$\mnorm{M_\delta-M_{\delta'}} < 2^{-n}$, while $\mnorm{T_\delta-T_{\delta'}} = 2/9$.
We next lift this to construct pairs $(U,\rho)$ giving a counterexample to robustness of the Schr\"odinger theory.

To motivate the construction, note that there is no unitary with the same support as~$M_\delta$ (if $n\geq3$).
This is because adjacent columns of the latter overlap in only one coordinate; hence it is not possible to choose phases to make them orthogonal.
We will therefore enlarge the dimension to lift the counterexample.
The key idea is to replace each entry of~$M_\delta$ by a multiple of the $2\times 2$ matrix~$J = \begin{psmallmatrix}1 & 1\\1 & 1\end{psmallmatrix}$.
Letting $\ket\pm = (\ket0 \pm \ket1)/\sqrt2$ denote the Hadamard basis states, we will lift the diagonal $J$ blocks to the rank-one projection $\proj+$, and the offdiagonal $J$ blocks to $\proj-$.
Because these two projections are mutually orthogonal and both have entrywise modulus~$J/2$, this \emph{almost} achieves the desired lift; we can complete the construction using one more dimension.

Formally, let $N=2n+1$ and consider the $N$-dimensional Hilbert space~$\C^N \cong ( \C^n \ot \C^2 ) \op \C \ket\bot$, with orthonormal standard (or ``computational'') basis
$\{\ket{x,0},\ket{x,1}:1\le x\le n\}\cup\{\ket{\bot}\}$.
For $0<\delta<1$, define a unitary~$U_\delta$ by
\begin{align*}
 U_\delta\ket{x,+}&=\ket{x,+} &&(1\le x\le n),\\
 U_\delta\ket{x,-}&=\ket{x-1,-} &&(2\le x\le n),\\
 U_\delta\ket{1,-}&=\delta\ket{n,-}+\sqrt{1-\delta^2}\ket{\bot},\\
 U_\delta\ket{\bot}&=\sqrt{1-\delta^2}\ket{n,-}-\delta\ket{\bot},
\end{align*}
where $\ket{x,\pm} = \ket x \ot \ket\pm = ( \ket{x,0}\pm\ket{x,1} ) / \sqrt2$.
The right-hand side vectors form an orthonormal basis, so $U_\delta$ is real orthogonal; in particular, it is unitary.
As the initial state we choose the uniform superposition of the first $2n$ coordinates:
\[
  \rho=\ket{\psi}\!\bra{\psi}, \qquad \ket{\psi}=\frac1{\sqrt n}\sum_{x=1}^n\ket{x,+}.
\]
Clearly, this state is fixed by the unitary~$U_\delta$, and the Born marginals are the uniform distribution on the first $2n$ coordinates.
The entrywise modulus of the $2n \times 2n$ submatrix obtained by restricting~$U_\delta$ to this ``active'' subspace is given by
\begin{align*}
    \abs{U_\delta}_{\text{act}} = \frac12 M_\delta \ot J.
\end{align*}
Indeed, each nonzero block receives exactly one contribution from either the plus or minus vector within a pair, and the corresponding rank-one projections have entrywise modulus proportional to~$J/2$, as noted above.
By \cref{lem:scale} and the uniqueness of the matrix scaling, we must have
\begin{equation}\label{eq:Pdelta}
    \PST(\rho, U_\delta)_{\text{act}} = \frac1{4n} T_\delta \ot J
\end{equation}
(the scaling factors lift to the two coordinates in each pair, and there is an overall factor~$1/(2n)$ because the Born marginals are uniform on the active subspace, rather than equal to the all-ones vector).
The last row and column of $\PST(\rho,U_\delta)$ are all zero, because the corresponding Born marginals vanish.
Thus we have fully determined the joint distribution assigned by Schr\"odinger theory to $\rho$ and $U_\delta$.

\begin{theorem}\label{thm:instability}
For every $n\ge2$, let $N=2n+1$, $\delta=2^{-n}$, and $\delta'=2^{-3n}$.
Then the real orthogonal matrices $U_\delta,U_{\delta'}$ have the same support (zero pattern w.r.t.\ the standard basis), both fix the state~$\rho$, and they satisfy
\begin{equation}\label{eq:main-separation}
 \mnorm{U_\delta-U_{\delta'}} < 2^{-n},\qquad
 \mnorm{\PST(\rho,U_\delta)-\PST(\rho,U_{\delta'})} > \frac1{10N}.
\end{equation}
Moreover, $\PST(\rho,U_\delta)$ and $\PST(\rho,U_{\delta'})$ have constant total-variation distance~$2/9$.
\end{theorem}
\begin{proof}
By \cref{lem:scale}, $t_\delta = \frac23$ and $t_{\delta'} = \frac89$.
Then \cref{eq:Pdelta} and the fact that the joint distributions are zero outside the active subspace yield
\[
  \mnorm{\PST(\rho,U_\delta)-\PST(\rho,U_{\delta'})}
= \frac {\mnorm{T_\delta - T_{\delta'}}} {4n}
= \frac {\abs{t_\delta - t_{\delta'}}} {4n}
= \frac1{18n}
> \frac1{10N}.
\]
Using \cref{eq:Pdelta}, we can also compute the total-variation distance:
\begin{align*}
    \dTV\mleft( \PST(\rho,U_\delta), \PST(\rho,U_{\delta'}) \mright)
= \frac1{2n} \sum_{j,i} \abs*{ (T_\delta)_{ji} - (T_{\delta'})_{ji} }
= \abs{t_\delta - t_{\delta'}}
= \frac29.
\end{align*}
It remains to verify the claimed properties of the unitaries and bound their distance.
Since $0<\delta,\delta'<1$, $U_\delta$ and $U_{\delta'}$ have the same support.
Moreover, they both fix the uniform superposition of the first $2n$ basis states.
Finally we verify the upper bound on their distance.
To this end, note that
\begin{align*}
    U_\delta - U_{\delta'}
= (\delta - \delta') \parens*{ \ketbra{n,-}{1,-}-\ketbra{\bot}{\bot} }
+ \parens*{ \sqrt{1-\delta^2}-\sqrt{1-{\delta'}^2} }
 \parens*{ \ketbra{\bot}{1,-} + \ketbra{n,-}{\bot} }.
\end{align*}
The four rank-one matrices in this formula have pairwise disjoint support, so we see that
\begin{align*}
    \mnorm{U_\delta - U_{\delta'}}
= \max \braces*{
    \abs{\delta - \delta'},
    \frac {\abs{\delta - \delta'}} 2,
    \frac {\abs{\sqrt{1-\delta^2}-\sqrt{1-{\delta'}^2}}} {\sqrt 2}
  }
= \abs{\delta - \delta'}
< 2^{-n},
\end{align*}
where the second equality uses that $\delta,\delta'\leq1/2$.
\end{proof}

\begin{remark}
In the above counterexample, the Born marginal of~$\rho$ assigns probability zero to~$\ket{\bot}$, and one might wonder whether robustness can be proven under the non-degeneracy condition that all marginals must be strictly positive.
This is not the case---it is possible to modify the above counterexample so that all marginals are strictly positive, e.g., by  adding a small maximally-mixed component to the density matrix: $\rho_\eps = (1-\eps) \rho + \eps I/N$ for suitable $\eps>0$.
\end{remark}

Now we prove our first main result, failure of robustness, which follows readily from the above.

\begin{proof}[Proof of \cref{thm:no go}]
Suppose, for contradiction, that the Schr\"odinger theory is block-robust.
Then, for $r(N) \coloneqq 10N$, there exists a positive polynomial $s(N)$ such that
\begin{align}\label{eq:contradict}
    \mnorm{U_\delta - U_{\delta'}} \leq \frac 1 {s(N)}
 \quad\Longrightarrow\quad
    \mnorm{\PST(\rho, U_\delta) - \PST(\rho, U_{\delta'})} \leq \frac 1 {10N}
\end{align}
for every $N=2n+1$, with $n\geq2$, and the choice $\delta=2^{-n}$ and $\delta'=2^{-3n}$ in \cref{thm:instability} (which also states that the two unitaries have the same support and hence the same minimal blocks).
But for $N=2n+1$ sufficiently large, we have $s(N) \leq 2^n$, and \cref{eq:contradict} is violated by the norm bounds in \cref{thm:instability}.
This is the desired contradiction.
\end{proof}

\section{A Capped Version of Schr\"odinger Theory}\label{sec:capped}
In this section we present a modification of the Schr\"odinger theory which satisfies symmetry, indifference, robustness and distributional product commutativity (\cref{thm:yay}).

For our \emph{capped Schr\"odinger theory}, we will first require some definitions.
For an $N \times N$ density matrix~$\rho$ and unitary~$U$, write
\begin{equation}\label{eq:p q C}
  p_i = \rho_{ii}, \quad q_j = (U \rho U^\dagger)_{jj}, \quad C_{ji} = 2 \abs{U_{ji}} \sqrt{q_j p_i}.
\end{equation}
The set of joint distributions with marginals $p$ and $q$, and entrywise upper bound~$C$ is denoted by
\[
    \qtransport(\rho, U)
\coloneqq \transport(p,q|C)
= \braces*{ P \in \RR_{\geq 0}^{N \times N} : \sum_j P_{ji} = p_i, \, \sum_i P_{ji} = q_j, \, P_{ji} \leq C_{ji} \text{ for all } j,i }.
\]
As explained above \cref{eq:transpoly}, we can also interpret the elements of $\qtransport(\rho,U)$ as flows on a bipartite network with supplies~$p$ and demands~$q$, and with capacity constraint~$C$.

\begin{lemma}[Feasibility]\label{lem:C-rho-u feasible}
For every density matrix~$\rho$ and every unitary~$U$, the set $\qtransport(\rho, U)$ is nonempty.
\end{lemma}
\begin{proof}
  Since~$p$ and~$q$ are both probability distributions, the max-flow min-cut theorem (\cref{thm:max-flow min-cut}) shows that it suffices to prove that, for every $I, J \subseteq [N]$,
\begin{align*}
    p(I) + q(J) \leq 1 + C(I,J),
\end{align*}
where $p(I) \coloneqq \sum_{i\in I} p_i$, $q(J) \coloneqq \sum_{j \in J} q_j$, and $C(I,J) \coloneqq \sum_{i \in I,j \in J} C_{ji}$.
To this end, let $P = \sum_{i \in I} \ket i\bra i$ and $Q = \sum_{j \in J} U^\dagger \ket j \bra j U$.
Then,
\begin{align*}
    p(I) + q(J) - 1 = \Tr \rho (P + Q - I).
\end{align*}
On the other hand,
\begin{align*}
    (P + Q - I)^2 = I - P - Q + PQ + QP
\end{align*}
is a positive semidefinite operator, hence $P + Q - I \preceq PQ + QP$ and
\begin{align*}
&\qquad  \Tr \rho (P + Q - I)
\leq \Tr \rho (PQ + QP)
= 2 \sum_{i \in I, j \in J} \operatorname{Re} \ \braket{j | U \rho | i} \braket{i | U^\dagger | j} \\
&\leq 2 \sum_{i \in I, j \in J} \abs{U_{ji}} \cdot \abs{\braket{j | U \sqrt\rho \sqrt\rho | i}}
\leq 2 \sum_{i \in I, j \in J} \abs{U_{ji}} \cdot \sqrt{\braket{j | U \rho U^\dagger | j}} \cdot \sqrt{\braket{i | \rho | i}}
= C(I, J),
\end{align*}
where the last inequality is the Cauchy--Schwarz inequality.
\end{proof}

\begin{definition}[Capped Schr\"odinger theory]
Given any density matrix $\rho$ and unitary~$U$, in any dimension~$N$, we define the joint distribution
\begin{align*}
  \PCST(\rho,U)
\coloneqq \argmin_{P \in \qtransport(\rho,U)} \KL(P \,\Vert\, \abs U)
= \argmin_{P \in \qtransport(\rho,U)} \sum_{\substack{j,i, \\ U_{ji} \neq 0}} P_{ji} \ln \frac {P_{ji}} {\abs{U_{ji}}}.
\end{align*}
This is well-defined because~$\qtransport(\rho,U)$ is nonempty (\cref{lem:C-rho-u feasible}), compact, and convex and the Kullback--Leibler divergence is strictly convex.
Then the \emph{capped Schr\"odinger theory} is defined by
\begin{equation}\label{eq:SCST via PCST}
  \SCST(\rho,U)_{ji} = \begin{cases}
    \PCST(\rho, U)_{ji}/p_i, & p_i>0, \\
    \abs{U_{ji}}^2,&p_i=0.
  \end{cases}
\end{equation}
This is a column stochastic matrix because each column is obtained either by conditioning the joint distribution~$\PCST(\rho,U)$ on an event with nonzero probability ($p_i>0$), or by taking the entrywise absolute value squared of a column of the unitary matrix~$U$.
We note that $\PCST(\rho,U) = \SCST(\rho,U) \diag(\rho_{11},\dots,\rho_{NN})$, so the former is indeed the joint distribution arising from the latter transition matrix.
\end{definition}

Note that in our definition of the Kullback--Leibler divergence we do not correct for the fact that the entries of~$\abs U$ will in general not form a probability distribution.
This is irrelevant for the definition of~$\PCST(\rho,U)$, since~$P$ is always a probability distribution and hence replacing $\abs U$ by its normalization only amounts to a constant additive shift by~$\ln \sum_{j,i} \abs U_{ji}$ of the objective.
Similarly, we can also replace the second argument of the Kullback--Leibler divergence by the capacity matrix~$C$, since
\begin{align}\label{eq:KL absU to KL C}
    \KL(P \,\Vert\, \abs U) - \KL(P \Vert C)
= \sum_{j,i} P_{ji} \ln (2 \sqrt{q_j p_i})
= \ln 2 - \frac {H(p) + H(q)} 2,
\end{align}
where $H(r) \coloneqq \sum_j r_j \ln(1/r_j)$ denotes the Shannon entropy of a probability distribution~$r$.
This replacement will be convenient for the robustness argument.

It is clear from the definition that the capped Schr\"odinger theory is symmetric and indifferent.
For the latter, note that $\PCST(\rho,U)_{ji} = 0 = \abs{U_{ji}}^2$ whenever $U_{ji}=0$, the former by the capacity constraint in the definition of~$\qtransport(\rho, U)$.
Next, we will establish distributional product commutativity.

\begin{lemma}[Distributional product commutativity]\label{lem:cst tensor identity}
Identify $\C^N \cong \C^A \otimes \C^B$ and $\R^N \cong \R^A \otimes \R^B$ using the same ordering of the product bases (e.g., the lexicographic ordering).
Let $\rho = \rho_A \ot \rho_B$ be a product state and $\mathcal U=U_A\otimes I_B$ and $\mathcal V=I_A\otimes V_B$ local unitaries.
\begin{enumerate}
\item[(a)] We have
  \begin{align*}
    \PCST(\rho_A \ot \rho_B, U_A \ot I_B) & = \PCST(\rho_A, U_A) \ot D_{\rho_B}, \\
    \PCST(\rho_A \ot \rho_B, I_A \ot V_B) & = D_{\rho_A} \ot \PCST(\rho_B, V_B),
  \end{align*}
  where $D_\sigma$ denotes the diagonal matrix with the same diagonal entries as $\sigma$.
\item[(b)] Consequently,
  \begin{align*}
      \SCST(\rho_A \ot \rho_B, U_A \ot I_B) (I_A \ot P_{\rho_B}) = \SCST(\rho_A, U_A) \ot P_{\rho_B}, \\
      \SCST(\rho_A \ot \rho_B, I_A \ot V_B) (P_{\rho_A} \ot I_B) = P_{\rho_A} \ot \SCST(\rho_B, V_B),
  \end{align*}
  where $P_\sigma = \sum_{k : \sigma_{kk}>0} \proj k$ is the projection onto the support of~$D_\sigma$.
\item[(c)] The capped Schr\"odinger theory satisfies distributional product commutativity.
\end{enumerate}
\end{lemma}
\begin{proof}
(a)~We only prove the first identity; the proof of the second identity is completely analogous.
For every~$P \in \qtransport(\rho_A \ot \rho_B, U_A \ot I_B)$, by the capacity constraint we have~$P_{j_A j_B, i_A i_B} = 0$ whenever~$j_B \neq i_B$, as the corresponding entry of~$\mathcal U$ is zero.
Therefore, we can regard~$P$ as a tuple of $A \times A$-matrices~$P^{(i_B)}$, satisfying the constraints
\begin{align*}
    \sum_{j_A} P_{j_A,i_A}^{(i_B)} = p_{i_A} r_{i_B}, \quad
    \sum_{i_A} P_{j_A,i_A}^{(i_B)} = q_{j_A} r_{i_B}, \quad
    0 \leq P_{j_A,i_A}^{(i_B)} \leq 2 \abs{(U_A)_{j_A,i_A}} \sqrt{q_{j_A} p_{i_A}} \, r_{i_B},
\end{align*}
where $p_{i_A}\coloneqq(\rho_A)_{i_A,i_A}$, $q_{j_A}\coloneqq(U_A\rho_AU_A^\dagger)_{j_A,j_A}$, and $r_{i_B}\coloneqq(\rho_B)_{i_B,i_B}$;
in other words, $P^{(i_B)} \in r_{i_B} \qtransport(\rho_A, U_A)$ for all~$i_B$.
The Kullback--Leibler divergence~$\KL(P \,\Vert\, \abs{\mathcal U})$ also decomposes as
\begin{align*}
    \KL(P \,\Vert\, \abs{\mathcal U})
= \sum_{i_B} \KL(P^{(i_B)} \,\Vert\, \abs{U_A}).
\end{align*}
Thus, minimizing the left-hand side over~$P \in \qtransport(\rho, \mathcal U)$ is equivalent to separately minimizing each of the right-hand side terms over~$P^{(i_B)} \in r_{i_B} \qtransport(\rho_A, U_A)$.
By definition, the unique minimizer of the left-hand side is $\PCST(\rho_A \ot \rho_B, U_A \ot I_B)$.
For the right-hand side terms there are two cases: if $r_{i_B}=0$, then we must have~$P^{(i_B)}=0$, while if~$r_{i_B}>0$ we have $\KL(P^{(i_B)} \,\Vert\, \abs{U_A}) = r_{i_B} \KL(P^{(i_B)}/r_{i_B} \Vert\, \abs{U_A}) + r_{i_B} \ln r_{i_B}$; in either case, the unique minimizer is given by~$r_{i_B} \PCST(\rho_A, U_A)$.
Therefore, $\PCST(\rho_A \ot \rho_B, U_A \ot I_B)^{(i_B)} = r_{i_B} \PCST(\rho_A, U_A)$, i.e., $\PCST(\rho_A \ot \rho_B, U_A \ot I_B) = \PCST(\rho_A, U_A) \ot D_{\rho_B}$.

(b)~This follows from part~(a) and \cref{eq:SCST via PCST}.

(c)~We first note that
\begin{align}\label{eq:p does not matter}
    P_{U_A \rho_A U_A^\dagger} \PCST(\rho_A,U_A) = \PCST(\rho_A,U_A).
\end{align}
Indeed, if $(U_A \rho_A U_A^\dagger)_{j_A,j_A}=0$ then we must have $\PCST(\rho_A,U_A)_{j_A,i_A} = 0$ for every~$i_A$, because the former is the~$j_A$-th row sum of the latter.
We verify \cref{eq:propercomm}:
\begin{align*}
    \SCST(\mathcal U\rho\mathcal U^\dagger,\mathcal V) \PCST(\rho,\mathcal U)
&=  \SCST(\mathcal U\rho\mathcal U^\dagger,\mathcal V) \bigl( \PCST(\rho_A,U_A) \ot D_{\rho_B} \bigr) \\
&=  \SCST(\mathcal U\rho\mathcal U^\dagger,\mathcal V) \bigl( P_{U_A \rho_A U_A^\dagger} \ot I_B  \bigr) \bigl( \PCST(\rho_A,U_A) \ot D_{\rho_B} \bigr) \\
&=   \bigl( P_{U_A \rho_A U_A^\dagger} \ot \SCST(\rho_B,V_B) \bigr) \bigl( \PCST(\rho_A,U_A) \ot D_{\rho_B} \bigr) \\
&=   P_{U_A \rho_A U_A^\dagger} \PCST(\rho_A,U_A) \ot \SCST(\rho_B,V_B) D_{\rho_B} \\
&=   \PCST(\rho_A,U_A) \ot \PCST(\rho_B,V_B)
\end{align*}
where the first equality follows from part~(a);
next we use \cref{eq:p does not matter} to add the projection;
the third equality follows from part~(b);
the penultimate equality is simply by rearranging;
and for the last equality we use $\PCST(\rho_B,V_B) = \SCST(\rho_B,V_B) D_{\rho_B}$ and again \cref{eq:p does not matter}.
An analogous calculation gives $\SCST(\mathcal V\rho\mathcal V^\dagger,\mathcal U) \PCST(\rho,\mathcal V) = \PCST(\rho_A,U_A) \ot \PCST(\rho_B,V_B)$.
\end{proof}

In the remainder of the section we establish robustness of the capped Schr\"odinger theory.
%
%
%
We first state a result that will later allow us to ``round'' elements of~$\qtransport(\rho, U)$ into~$\qtransport(\tilde \rho, \tilde U)$:

\begin{lemma}\label{thm:transportation correction}
Let $A \in \RR_{\geq0}^{N \times N}$ and let $r,s\in\RR_{\geq0}^N$ such that the transportation polytope~$\transport(r,s|A)$ defined in~\eqref{eq:transpoly} is nonempty.
Let $X \in \RR_{\geq0}^{N \times N}$ be such that $X_{ji} \leq A_{ji}$ for all~$i,j\in[N]$.
Then there exists~$Y \in \transport(r,s|A)$ such that
\begin{align*}
    \sum_{i,j} \abs{Y_{ji} - X_{ji}} \leq N \left( \sum_i \abs*{ \sum_j X_{ji} - r_i } + \sum_j \abs*{ \sum_i X_{ji} - s_j } \right).
\end{align*}
\end{lemma}
\begin{proof}
Let $Y \in \transport(r,s|A)$ be the element of minimal $\ell^2$-distance%
\footnote{Alternatively one can use a minimizer for the $\ell^1$-distance.}
to~$X$, and set~$Z \coloneqq Y - X$.
We can use~$Z$ to define a directed bipartite graph~$G$ with both vertex sets indexed by~$[N]$: 
there is an edge between left vertex~$j$ and right vertex~$i$ if and only if~$Z_{ji} \neq 0$; it is oriented from left to right if~$Z_{ji} < 0$ and otherwise from right to left.
We claim that~$G$ is acyclic.
Suppose for contradiction that there is a directed simple cycle~$i_1 \to j_1 \to i_2 \to j_2 \to \ldots \to i_\ell \to j_\ell \to i_{\ell+1} = i_1$.
Define~$\Delta \coloneqq \sum_{m=1}^\ell E_{j_m,i_m} - E_{j_m,i_{m+1}}$, where~$E_{j,i} \in \RR^{N \times N}$ denotes the elementary matrix with a single entry in row~$j$ and column~$i$.
Choose $0 < \eps < \min_{\text{cycle}} \abs{Z_{ji}}$.
For entries on the cycle, if $Z_{ji}>0$, then $0 < (\eps\Delta)_{ji} < Z_{ji}$, while if $Z_{ji}<0$ we have $Z_{ji} < (\eps\Delta)_{ji} < 0$.
Thus, each altered entry of $Y - \eps \Delta$ lies between the corresponding entries of~$X$ and~$Y$, so it remains in~$[0,A_{ji}]$, and it moves strictly closer to~$X_{ji}$.
Since the row and column sums are also unchanged, this contradicts the minimality of~$Y$.

Because~$G$ is acyclic, we may choose a topological ordering~$v_1 < \ldots < v_{2N}$ of its vertices, i.e., there is an edge~$v_a \to v_b$ only if $a < b$.
We define a flow~$f$ on~$G$ by assigning $\abs{Z_{ji}}$ to the unique edge between left vertex $j$ and right vertex $i$ (if $Z_{ji}\neq0$).
Then:
\begin{align*}
    \sum_{j,i} \abs{Y_{ji} - X_{ji}}
= \sum_{j,i} \abs{Z_{ji}}
= \sum_e f(e)
\leq \sum_{v_a \to v_b} (b-a) \, f(v_a \to v_b)
= \sum_{b=1}^{2N} b \, \delta(v_b)
= \sum_{b=1}^{2N} \left( b - \frac{2N+1}2 \right) \, \delta(v_b),
\end{align*}
where $\delta(v) \coloneq \sum_{u : u \to v} f(u \to v) - \sum_{w : v \to w} f(v \to w)$ denotes the inflow minus outflow at vertex~$v$;
the last equality holds because $\sum_v \delta(v) = 0$.
Therefore:
\begin{align*}
    \sum_{j,i} \abs{Y_{ji} - X_{ji}}
\leq \frac{2N-1}2 \sum_v \abs{\delta(v)}
= \frac{2N-1}2 \left( \sum_i \abs*{ \sum_j Z_{ji} } + \sum_j \abs*{ \sum_i Z_{ji} } \right),
\end{align*}
and now the claim follows from $Z = Y - X$ and $Y \in \transport(r,s|A)$.
\end{proof}

We will also require some elementary continuity estimates.

\begin{lemma}\label{lem:elementary estimates}
\begin{enumerate}
\item[(a)] For $h(x) \coloneqq x \ln(1/x)$ on $[0,1]$ (with $h(0)=0$), one has $\abs{h(x)-h(y)} \leq \frac{2}{e} \sqrt{\abs{x-y}}$.
\item[(b)] The function $\eta_c(x) \coloneqq x \ln(x/c)$ on $[0,c]$ (with $\eta_0(0)=0$) satisfies
  \[
     \abs{\eta_c(x) - \eta_c(y)} \leq \frac{2}{e} \sqrt{c \abs{x-y}}
     \qquad (0 \leq x,y \leq c).
  \]
\item[(c)] For $0 \leq x \leq c$ and $0 \leq y \leq c'$, we have
  \[
     \abs{\eta_c(x) - \eta_{c'}(y)} \leq \frac2e \sqrt{c \abs{x-y}} + \abs{c-c'}.
  \]
\item[(d)] Thus, for any two states $\rho,\tilde\rho$, unitaries~$U,\tilde U$, and $Q \in \qtransport(\rho,U)$ and $\tilde Q \in \qtransport(\tilde\rho,\tilde U)$, we have
  \begin{align*}
        \abs*{ \KL(Q \Vert C) - \KL(\tilde Q \Vert \tilde C) }
    \leq \frac{2\sqrt2}e N^{1/4} \sqrt{\norm{Q - \tilde Q}_1} + \norm{C - \tilde C}_1,
  \end{align*}
  where $C,\tilde C$ denote the capacity matrices associated with $(\rho,U)$ and $(\tilde\rho,\tilde U)$, respectively, as in~\cref{eq:p q C}.
\end{enumerate}
\end{lemma}
\begin{proof}
(a)~We may assume without loss of generality that $x > y$.
    By concavity and using $h(0)=0$, we have $h(y) \geq \frac y x h(x)$ and $h(x-y) \geq \frac{x-y} x h(x)$, hence
    \begin{align*}
        h(x) - h(y)
    \leq \frac{x - y} x h(x)
    \leq h(x-y)
    \leq \sqrt{x-y} \max_{t \in [0,1]} \sqrt t \ln(1/t)
    = \frac2e \sqrt{x-y},
    \end{align*}
    as the right-hand side is maximized at $t=e^{-2}$.
    On the other hand, $h' \geq -1$ and $0 \leq h \leq h(e^{-1}) = e^{-1}$ on~$[0,1]$, hence
    \begin{align*}
        h(x) - h(y) \geq -\min \{ x-y, \frac1e \} \geq -\sqrt{\frac{x-y}e} \geq -\frac2e\sqrt{x-y}.
    \end{align*}

(b)~This follows at once from part~(a), because $\eta_c(x) = -c \, h(x/c)$ for $c>0$ and $x/c \in [0,1]$ (the case $c=0$ is immediate).

(c)~Let $z \coloneqq \min(y,c)$.
    Then,
    \begin{align*}
        \abs{\eta_c(x) - \eta_c(z)} \leq \frac2e \sqrt{c \abs{x-z}} \leq \frac2e \sqrt{c \abs{x-y}}
    \end{align*}
    by part~(b), and it suffices to prove that
    \begin{align*}
        \abs{\eta_c(z) - \eta_{c'}(y)} \leq \abs{c - c'}.
    \end{align*}
    Indeed, if $y>c$ then $z=c$ and $\eta_c(z)=0$, hence
    \begin{align*}
        \abs{\eta_c(z) - \eta_{c'}(y)}
    = \abs{\eta_{c'}(y)}
    = y \ln(c'/y)
    \leq c' - y \leq c' - c
    \end{align*}
    using the concavity of the logarithm, while if $y \leq c$ then $z = y$ and, similarly,
    \begin{align*}
        \abs{\eta_c(z) - \eta_{c'}(y)}
    =   \abs{\eta_c(y) - \eta_{c'}(y)}
    =   y \abs{\ln(c'/c)}
    \leq \abs{c'-c}
    \end{align*}
    (here we assume $y>0$ so that $c,c'>0$, otherwise both functions vanish).
    This concludes the proof of part~(c).

(d)~Because $Q_{ji} \leq C_{ji}$ and $\tilde Q_{ji} \leq \tilde C_{ji}$, we have
    \begin{align*}
        \abs*{ \KL(Q\Vert C) - \KL(\tilde Q\Vert \tilde C) }
    \leq \frac2e \sum_{j,i} \sqrt{ C_{ji} \abs{Q_{ji} - \tilde Q_{ji}}} + \norm{C - \tilde C}_1
    \leq \frac2e \sqrt{\sum_{j,i} C_{ji}} \sqrt{\norm{Q - \tilde Q}_1} + \norm{C - \tilde C}_1
    \end{align*}
    using part~(c) and the Cauchy--Schwarz inequality.
    Because $U$ is a unitary and $p$, $q$ are probability distributions, again using the Cauchy--Schwarz inequality,
    \begin{align*}
        \sum_{j,i} C_{ji}
    = 2 \sum_{j,i} \abs{U_{ji}} \sqrt{q_j p_i}
    \leq 2 \sqrt{\sum_{j,i} \abs{U_{ji}}^2} \sqrt{\sum_{j,i} q_j p_i}
    \leq 2 \sqrt N,
    \end{align*}
    and now the bound follows.
\end{proof}

We now establish robustness of the capped Schr\"odinger theory:

\begin{theorem}[Robustness]\label{thm:cst quantitative robustness}
Let $\delta\geq0$.
Suppose~$(\rho, U)$,~$(\tilde\rho, \tilde U)$ satisfy~$\mnorm{\rho - \tilde\rho}, \mnorm{U - \tilde U} \leq \delta$.
Then:
\[
    \mnorm{\PCST(\rho,U) - \PCST(\tilde\rho, \tilde U)}
\leq \frac12\norm{\PCST(\rho,U) - \PCST(\tilde\rho, \tilde U)}_1
\leq 5 N \delta^{1/8}.
\]
In particular, the capped Schr\"odinger theory is robust.
\end{theorem}
\begin{proof}
Since both $\PCST(\rho,U), \PCST(\tilde\rho, \tilde U)$ are probability distributions by definition, the bound is trivial if $5 N \delta^{1/8} \geq 1$.
So without loss of generality we may assume that $\delta\leq1/(5 N)^8\leq1$.
First, we establish $\ell^1$-distance bounds between the Born marginals and the capacity matrices associated with~$(\rho,U)$ and $(\tilde\rho,\tilde U)$, respectively.
Let $p$ and $\tilde p$ denote the Born marginals of $\rho$ and $\tilde \rho$.
Then:
\begin{align*}
    \abs{p_i - \tilde p_i} = \abs{ \braket{i | \rho - \tilde\rho | i} } \leq \mnorm{\rho - \tilde\rho} \leq \delta \quad (\forall i).
\end{align*}
Next, we bound the distance between the Born marginals $q$ and $\tilde q$ of $U \rho U^\dagger$ and $\tilde U \tilde\rho \tilde U^\dagger$:
\begin{align*}
    \abs{q_j - \tilde q_j}
= \abs{ \braket{j | U \rho U^{\dagger} - \tilde U \tilde\rho \tilde U^\dagger | j} }
\leq \norm{U \rho U^\dagger - \tilde U \tilde\rho \tilde U^\dagger}_{\ope}
\leq 2 \norm{U - \tilde U}_{\ope} + \norm{\rho - \tilde\rho}_{\ope}
\leq 3 N \delta,
\end{align*}
because $\norm{X}_{\ope} \leq N \mnorm{X}$ for any $N\times N$ matrix~$X$.
Finally, we bound the distance between the capacity matrices~$C, \tilde C$ defined as in \cref{eq:p q C} for the pairs $\rho,U$ and $\tilde\rho,\tilde U$, respectively:
\begin{align*}
    \abs{C_{ji} - \tilde C_{ji}}
&= 2 \abs*{ \abs{U_{ji}} \sqrt{q_j p_i} - \abs{\tilde U_{ji}} \sqrt{\tilde q_j \tilde p_i} } \\
&\leq 2 \abs*{ \abs{U_{ji}} - \abs{\tilde U_{ji}} } \sqrt{q_j p_i} + 2 \abs{\tilde U_{ji}} \abs*{ \sqrt{q_j p_i} - \sqrt{\tilde q_j \tilde p_i} } \\
&\leq 2 \delta + 2 \sqrt{\abs*{ q_j p_i - \tilde q_j \tilde p_i} } \\
&\leq 2 \delta + 2 \sqrt{\abs*{ q_j - \tilde q_j} p_i + \tilde q_j \abs*{ p_i - \tilde p_i} } \\
&\leq 2 \delta + 2 \sqrt{3 N \delta + \delta}
\leq 6 \sqrt{N \delta}.
\end{align*}
Thus:
\begin{align*}
    \norm{p - \tilde p}_1 \leq N\delta \leq \sqrt\delta, \qquad
    \norm{q - \tilde q}_1 \leq 3N^2\delta \leq \sqrt\delta, \qquad
    \norm{C - \tilde C}_1 \leq 6 N^{5/2} \sqrt\delta \eqqcolon \kappa.
\end{align*}

Now set~$P \coloneqq \PCST(\rho,U)$ and define $P'_{ji} \coloneqq \min(P_{ji}, \tilde C_{ji})$.
Since $P_{ji} \leq C_{ji}$, we have $\norm{P' - P}_1 \leq \kappa$.
Denoting the row and column sums of~$P'$ by~$q'$ and $p'$, respectively, it follows that $\norm{p' - p}_1, \norm{q' - q}_1 \leq \kappa$; hence we have $\norm{p' - \tilde p}_1, \norm{q' - \tilde q}_1 \leq \sqrt\delta + \kappa$.
Thus, \cref{thm:transportation correction} shows that there exists~$\tilde Q \in \qtransport(\tilde\rho,\tilde U)$ such that
$\norm{P' - \tilde Q}_1
    \leq N \left( \norm{p' - \tilde p}_1 + \norm{q' - \tilde q}_1 \right)
    \leq 2 N (\sqrt\delta + \kappa)
    \leq 14 N^{7/2} \sqrt\delta$,
and hence
\[ \norm{\tilde Q - P}_1 \leq \norm{\tilde Q - P'}_1 + \norm{P' - P}_1 \leq 14 N^{7/2} \sqrt\delta + \kappa \leq 20 N^{7/2} \sqrt\delta. \]
Likewise, for $\tilde P \coloneqq \PCST(\tilde\rho,\tilde U)$, there exists $Q \in \qtransport(\rho,U)$ such that
\[ \norm{Q - \tilde P}_1 \leq 20 N^{7/2} \sqrt\delta. \]
Recall from \cref{eq:KL absU to KL C} that we can also define $P$ as the minimizer of~$\KL(\cdot\Vert C)$ over $\qtransport(\rho,U)$, and likewise for~$\tilde P$.
We show that $Q$ is an approximate minimizer for this objective:
\begin{align}
\nonumber
&\quad \KL(Q \Vert C) - \KL(P \Vert C) \\
\nonumber
&= [ \KL(Q \Vert C) - \KL(\tilde P \Vert \tilde C) ]
+ \underbrace{[ \KL(\tilde P \Vert \tilde C) - \KL(\tilde Q \Vert \tilde C) ]}_{\leq0}
+ [ \KL(\tilde Q \Vert \tilde C) - \KL(P \Vert C) ] \\
\nonumber
&\leq \frac{2\sqrt2}e N^{1/4} \left( \sqrt{\norm{Q - \tilde P}_1} + \sqrt{\norm{\tilde Q - P}_1} \right) + 2 \norm{C - \tilde C}_1 \\
\label{eq:approx min}
&\leq \frac{8\sqrt{10}}e N^2 \delta^{1/4} + 12 N^{5/2} \sqrt\delta
\leq 10 N^2 \delta^{1/4},
\end{align}
where we used \cref{lem:elementary estimates}~(d) to upper bound the first and the last term; the underbraced middle term is nonpositive because~$\tilde P$ is a minimizer for $\KL(\cdot \Vert \tilde C)$ over $\qtransport(\tilde\rho,\tilde U)$;
the final estimate uses our running assumption that $\delta \leq 1/(5N)^8$.

On the other hand, using the fact that $P$ is a minimizer on $\qtransport(\rho,U)$ and the strong convexity estimate~\eqref{eq:strong convex} for $P_t \coloneqq (1-t) P + t Q$ ($t \in (0,1)$):
\begin{align*}
     \KL(Q \Vert C) - \KL(P \Vert C)
\geq \frac1t \Bigl( (1-t) \KL(P \Vert C) + t \KL(Q \Vert C) - \KL(P_t \Vert C) \Bigr)
\geq \frac{1-t}2 \norm{P - Q}_1^2.
\end{align*}
Letting $t \downarrow 0$, we obtain $\KL(Q \Vert C) - \KL(P \Vert C) \geq \frac12 \norm{P - Q}_1^2$, and hence $\norm{P - Q}_1 \leq \sqrt{20} N \delta^{1/8}$, using \cref{eq:approx min}.
Thus:
\begin{align*}
    \norm{P - \tilde P}_1
\leq \norm{P - Q}_1 + \norm{Q - \tilde P}_1
\leq \sqrt{20} N \delta^{1/8} + 20 N^{7/2} \sqrt\delta
\leq 5 N \delta^{1/8}
\end{align*}
and now the claim follows from the relation $\mnorm{P - \tilde P} \leq \dTV(P,\tilde P) = \frac12\norm{P-\tilde P}_1$.
\end{proof}

\begin{proof}[Proof of \cref{thm:yay}]
As discussed above, symmetry and indifference are clear from the definition.
Distributional product commutativity is \cref{lem:cst tensor identity}~(c).
Robustness is proved in \cref{thm:cst quantitative robustness}.
\end{proof}

\section{Limitations on Axioms of Hidden-Variable Theories}\label{sec:limitations}

\subsection{Indifference and Product Commutativity are Incompatible}\label{sec:productcom}

\begin{proof}[Proof of \cref{thm:productcom}]
We will assume both indifference and product commutativity to obtain a contradiction in dimension~$N=12$.
Let us denote the standard basis of~$\C^{12}$ by $\ket0$, \dots, $\ket{11}$.
We will consider the pure state~$\rho = \proj\psi$, where $\ket\psi=(\ket0 + \ket1 + \ket2 - \ket3)/2$, which is a product state with respect to two distinct tensor product factorizations: $\C^{12} = \C^2 \ot \C^6$ and $\C^{12} = \C^3 \ot \C^4$.
Likewise, the unitary~$\mathcal V = I_6 \ot H$, where $H$ is the Hadamard gate, is a local unitary with respect to either factorization.
By indifference, for any state~$\sigma$, the transition matrix $S(\sigma, \mathcal V)$ is block diagonal with $2\times2$ diagonal blocks, i.e., it preserves the span of $\{\ket{2k},\ket{2k+1}\}$ for~$k=0,\dots,5$.
In particular, $S(\rho, \mathcal V) \ket 6 = p \ket 6 + (1-p) \ket 7$ for some~$p\in[0,1]$.
The key idea will be to use product commutativity twice, once for each tensor product decomposition, to derive contradictory values of~$p$.

We first consider the factorization $\C^{12} = \C^2 \ot \C^6$ and the local unitary $\mathcal U = H \ot I_6$.
By product commutativity, we must have
\begin{align}\label{eq:prod comm 1}
    S(\mathcal U \rho \mathcal U^\dagger, \mathcal V) S(\rho, \mathcal U)
= S(\mathcal V \rho \mathcal V^\dagger, \mathcal U) S(\rho, \mathcal V).
\end{align}
Now, $\mathcal U\ket\psi = (\ket0 + \ket1 + \ket2 - \ket3 + \ket6 + \ket7 + \ket8 - \ket9)/\sqrt8$ and $\mathcal V\mathcal U\ket\psi = (\ket0 + \ket3 + \ket 6 + \ket 9)/2$.
The transition matrix $S(\mathcal U \rho \mathcal U^\dagger, \mathcal V)$ has to map the Born marginals of the former to the Born marginals of the latter.
By indifference, it follows that $S(\mathcal U \rho \mathcal U^\dagger, \mathcal V) \ket0 = \ket0$ and $S(\mathcal U \rho \mathcal U^\dagger, \mathcal V) \ket6 = \ket6$.
Also by indifference, for any state~$\sigma$, $S(\sigma, \mathcal U)$ preserves the span of $\{\ket k, \ket{k+6}\}$ for~$k=0,\dots,5$.
It follows that $S(\mathcal U \rho \mathcal U^\dagger, \mathcal V) S(\rho, \mathcal U) \ket6$
is a probability distribution supported on~$\{0,6\}$.
On the other hand:
\begin{align*}
    S(\mathcal V \rho \mathcal V^\dagger, \mathcal U) S(\rho, \mathcal V) \ket6
= p S(\mathcal V \rho \mathcal V^\dagger, \mathcal U) \ket6
+ (1-p) S(\mathcal V \rho \mathcal V^\dagger, \mathcal U) \ket7.
\end{align*}
Again by indifference, this is a probability distribution with probability mass~$p$ on~$\{0,6\}$ and probability mass~$1-p$ on~$\{1,7\}$.
By \cref{eq:prod comm 1}, we must have $p=1$.

Now we consider the factorization $\C^{12} = \C^3 \ot \C^4$ and the local unitary $\mathcal U = (H \op 1) \ot I_4$.
By product commutativity, we must have
\begin{align}\label{eq:prod comm 2}
    S(\mathcal U \rho \mathcal U^\dagger, \mathcal V) S(\rho, \mathcal U)
= S(\mathcal V \rho \mathcal V^\dagger, \mathcal U) S(\rho, \mathcal V).
\end{align}
This time around, $\mathcal U\ket\psi = (\ket0 + \ket1 + \ket2 - \ket3 + \ket4 + \ket5 + \ket6 - \ket7)/\sqrt8$ and $\mathcal V\mathcal U\ket\psi = (\ket0 + \ket3 + \ket4 + \ket7)/2$.
The transition matrix $S(\mathcal U \rho \mathcal U^\dagger, \mathcal V)$ has to map the Born marginals of the former to the Born marginals of the latter.
By indifference, it follows that $S(\mathcal U \rho \mathcal U^\dagger, \mathcal V) \ket2 = \ket3$ and $S(\mathcal U \rho \mathcal U^\dagger, \mathcal V) \ket6 = \ket7$.
Also by indifference, for any state~$\sigma$, $S(\sigma, \mathcal U)$ preserves the span of~$\{\ket k,\ket{k+4}\}$ for $k=0,\dots,3$ as well as the span of~$\{\ket k\}$ for $k=8,\dots,11$.
Thus, $S(\mathcal U \rho \mathcal U^\dagger, \mathcal V) S(\rho, \mathcal U) \ket6$ is a probability distribution supported on~$\{3,7\}$.
On the other hand:
\begin{align*}
    S(\mathcal V \rho \mathcal V^\dagger, \mathcal U) S(\rho, \mathcal V) \ket6
= p S(\mathcal V \rho \mathcal V^\dagger, \mathcal U) \ket6
+ (1-p) S(\mathcal V \rho \mathcal V^\dagger, \mathcal U) \ket7.
\end{align*}
Again by indifference, this is a probability distribution with probability mass~$p$ on $\{2,6\}$ and probability mass~$1-p$ on~$\{3,7\}$.
By \cref{eq:prod comm 2}, we must have $p=0$.
This is the desired contradiction.
\end{proof}

\subsection{Indifference and Decomposition Invariance are Incompatible}

\begin{proof}[Proof of \cref{thm:decompinvcom}]
Take $U = H \op 1$ on $\C^3$, where $H$ is the Hadamard gate, so
\begin{align*}
    U = \begin{pmatrix}1/\sqrt2 & 1/\sqrt2 & 0 \\ 1/\sqrt2 & -1/\sqrt2 & 0 \\ 0 & 0 & 1\end{pmatrix}.
\end{align*}
Clearly, the minimal blocks are $(\{1,2\},\{1,2\})$ and $(\{3\},\{3\})$, hence indifference implies that the stochastic matrix must look like:
\begin{align*}
    S(\rho, U) = \begin{pmatrix}\star & \star & 0 \\ \star & \star & 0 \\ 0 & 0 & 1\end{pmatrix}
\end{align*}
Let $\ket\pm = (\ket0 \pm \ket1)/\sqrt2$; here and below we enumerate the standard basis of~$\C^3$ as $\ket0,\ket1,\ket2$.
Then, $H\ket+=\ket0$ and $H\ket-=\ket1$, hence compatibility with the Born marginals of the final states forces
\begin{align*}
    S(\proj+, U) = \begin{pmatrix}1 & 1 & 0 \\ 0 & 0 & 0 \\ 0 & 0 & 1\end{pmatrix} \quad\text{and}\quad
    S(\proj-, U) = \begin{pmatrix}0 & 0 & 0 \\ 1 & 1 & 0 \\ 0 & 0 & 1\end{pmatrix},
\end{align*}
(the stochasticity of $S$ forces the nonzero entries to be one).
Likewise, the same argument shows that for the initial states $\rho_\pm \coloneqq \frac12 \proj{\pm} + \frac12 \proj2$,
\begin{align*}
    S(\rho_+, U) = \begin{pmatrix}1 & 1 & 0 \\ 0 & 0 & 0 \\ 0 & 0 & 1\end{pmatrix} \quad\text{and}\quad
    S(\rho_-, U) = \begin{pmatrix}0 & 0 & 0 \\ 1 & 1 & 0 \\ 0 & 0 & 1\end{pmatrix}.
\end{align*}
On the other hand, decomposition invariance yields
\begin{align*}
    S(\rho_\pm, U) = \frac12 S(\proj\pm, U) + \frac12 S(\proj2, U)
\end{align*}
and hence we must have
\begin{align*}
    S(\proj+, U) = S(\proj2, U) = S(\proj-, U),
\end{align*}
a contradiction.
\end{proof}

\begin{remark}
  A closely related three-dimensional construction is used in~\cite[Sec.~IV]{Aaronson} to show that for any indifferent hidden-variable theory, the assignment~$(\rho, U) \mapsto S(\rho,U)$ cannot be continuous in~$\rho$.
  For our unitary, the analogous pure-state construction would be to use~$\ket{\psi_\pm(\delta)} = \delta \ket\pm + \sqrt{1 - \delta^2} \ket2$.
  Alternatively, one can use the mixed states~$\rho_{\pm}(\delta) = \delta \proj\pm + (1-\delta) \proj 2$ for~$\delta\in(0,1)$.
  Both versions give the claimed discontinuity: $S(\rho_{\pm}(\delta), U)$ is given by the same expressions as in the proof of \cref{thm:decompinvcom}; therefore the two transition matrices have different limits as~$\delta \to 0$, yet~$\rho_{\pm}(\delta) \to \proj2$.
\end{remark}

\section*{Acknowledgments}
GM is supported by the European Research Council through an ERC Starting Grant (Grant agreement No.~101077455, ObfusQation) and funded by the Deutsche Forschungsgemeinschaft (DFG, German Research Foundation) under Germany's Excellence Strategy - EXC 2092 CASA - 390781972.
HN is supported by the European Union's Horizon Europe research and innovation programme under the Marie Sk\l{}odowska-Curie Actions (MSCA) Postdoctoral Fellowship, Grant Agreement No. 101212204 (AsympTensorPolytope), as well as Villum Fonden via the QMATH Centre of Excellence (Grant No.~10059).
MW acknowledges support by the European Union (ERC Grant SYMOPTIC, 101040907), by the Deutsche Forschungsgemeinschaft (DFG, German Research Foundation, 556164098), by the Deutsche Forschungsgemeinschaft under Germany's Excellence Strategy~--~EXC-2111~--~390814868, by the German Federal Ministry of Research, Technology and Space (QuSol, 13N17173), and by the Klaus Tschira Foundation.
He also thanks the Simons Institute for the Theory of Computing at UC Berkeley and Q-FARM and the Leinweber Institute for Theoretical Physics at Stanford University for hospitality.
Views and opinions expressed are those of the author(s) only and do not necessarily reflect those of the European Union or the European Research Council Executive Agency. Neither the European Union nor the granting authority can be held responsible for them.

\bibliographystyle{alpha}
\bibliography{references}

\appendix
\crefalias{section}{appendix}

\end{document}